\documentclass[graybox]{svmult}
\usepackage{mathptmx}       
\usepackage{helvet}         
\usepackage{courier}        
\usepackage{makeidx}         
\usepackage{graphicx}        
\usepackage{multicol}        
\usepackage[bottom]{footmisc}
\usepackage{amsmath,amssymb,amsfonts}        

\usepackage{lipsum}
\usepackage{epstopdf}
\usepackage{algorithm}
\usepackage{algorithmic}
\usepackage[utf8]{inputenc}
\usepackage{amsmath}
\usepackage{latexsym}
\usepackage{subfigure}
\usepackage{stmaryrd}
\usepackage{caption}
\usepackage{xcolor}
\usepackage[colorinlistoftodos]{todonotes}
\usepackage{enumerate}
\usepackage{bbm}
\usepackage{bm}
\usepackage{xspace,comment}
\usepackage{hyperref}

\usepackage{amsopn}

\newcommand{\QED}{\Box} 
    
\newcommand{\Real}{\mathbb{R}}  
\newcommand{\Ind}{\mathbbm{1}}

\newcommand\dtv[1]{\left\| #1 \right\|_{\rm tv}}

\newcommand{\mB}{{\mathcal B}}
\newcommand{\mP}{{\mathcal P}}

\newcommand{\mC}{{\mathcal C}}

\newcommand{\mF}{{\mathcal F}}

\newcommand{\mX}{{\mathcal X}}
\newcommand{\mY}{{\mathcal Y}}
\newcommand{\mN}{{\mathcal N}}

\newcommand{\mbP}{\mathbb{P}}   
\newcommand{\mbS}{\mathbb{S}} 
\newcommand{\mbE}{\mathbb{E}} 
\newcommand{\mbR}{\mathbb{R}}

\newcommand{\sd}{{\sf d}}

\newcommand{\sS}{{\sf S}}

\newcommand{\X}{\mathtt{X}}

\newcommand{\W}{\mathtt{W}}

\newcommand{\beq}{\begin{equation}}
\newcommand{\eeq}{\end{equation}}
\newcommand{\beqa}{\begin{eqnarray}}
\newcommand{\eeqa}{\end{eqnarray}}
\newcommand{\nn}{\nonumber}

\newcommand{\dy}{d_y}
\newcommand{\dx}{d_x}

\newtheorem{assumption}{Assumption}

\makeindex             

\begin{document}
\title*{On a class of constrained particle filters for continuous-discrete state space models}
\titlerunning{Constrained particle filters for continuous-discrete state space models}
\author{Utku Erdogan, Gabriel J. Lord and Joaquin Miguez}
\institute{Utku Erdogan \at Eskisehir Technical University (Turkey), \email{utkuerdogan@eskisehir.edu.tr}
\and Gabriel J. Lord \at Radboud University (The Netherlands), \email{gabriel.lord@ru.nl} \and Joaquin Miguez \at Universidad Carlos III de Madrid (Spain), \email{joaquin.miguez@uc3m.es}}
%
%
\maketitle

\abstract{
Particle filters (PFs) are recursive Monte Carlo algorithms for Bayesian tracking and prediction in state space models. This paper addresses continuous‑discrete filtering problems, where the hidden state evolves as an Itô stochastic differential equation (SDE) and observations arrive at discrete times. We propose a novel class of constrained PFs that enforce compact support on the state at each observation instant, thereby limiting exploration to plausible regions of the state space. Unlike earlier approaches that truncate the likelihood, the proposed method constrains the dynamics directly, yielding improved numerical stability. Under standard regularity assumptions, we prove convergence of the constrained filter, derive uniform‑in‑time error estimates, and extend the analysis to account for discretisation errors arising from numerical SDE solvers. A numerical study on a stochastic Lorenz‑96 system demonstrates the practical application of the methodology when the constraint is implemented via barrier functions.
}

%
\section{Introduction}

%
\subsection{Motivation and background}

Bayesian filtering is a probabilistic methodology for the processing of sequentially-observed data. The filtering problem involves the recursive computation of the probability distribution of the state of a dynamical system conditional on observations collected over time \cite{Anderson79,Bain08}, and applications include data assimilation \cite{Law15, Reich15}, optimal control \cite{Mitter02,Georgiou13}, nonlinear tracking \cite{Ristic04,Stone13}, and others. 

The exact implementation of Bayesian filters is unfeasible beyond a few examples, the most relevant being the Kalman filter for linear and Gaussian models \cite{Kalman60,Kalman61,Anderson79,Ristic04}. In practice, approximation techniques are employed, including Kalman-based methods \cite{Anderson79,Evensen03,Julier04}, recursive Monte Carlo algorithms \cite{Gordon93,Doucet00,Evensen03,Djuric03} or variational schemes \cite{Smidl08,Frazier23}. These techniques suffer from limitations which are related to a trade-off involving stability, accuracy and computational cost. For example, Kalman-based algorithms perform poorly with strongly nonlinear systems, or when noise distributions are heavy-tailed \cite{Ristic04}. Particle filters (PFs) \cite{Gordon93,Kitagawa96,Doucet00,Djuric03} are sequential Monte Carlo methods that handle nonlinearities in a rather natural way, but may collapse when the observations are highly-informative or the model dimension is high \cite{Bengtsson08,Snyder08}.

In this work, we are interested in a class of continuous-discrete filtering problems, where 
\begin{itemize}
\item the system state is a continuous-time stochastic process $\X(t)$, $t \in [0,T]$, characterised by an It\^o stochastic differential equation (SDE), and 
\item the observations $Y_1, \ldots, Y_M$ are collected sequentially at specific time instants $t_1<t_2<\cdots<t_M$, with $t_1 > 0$ and $t_M \le T$.
There are many real-world problems that can be naturally modelled within this framework. See \cite{Akyildiz24} for further discussion.
\end{itemize}

The continuous-discrete setting exacerbates the difficulties to devise efficient PFs. If the underlying SDE characterising the state is nonlinear, one can simulate it using numerical schemes \cite{Kloeden95} and this is sufficient to implement plain PFs \cite{Akyildiz24}. However, transition probability densities (from $t_{n-1}$ to $t_n$) cannot be evaluated in general, which makes the design of proposal distributions for sequential importance sampling \cite{Doucet00} very hard. One practical alternative is to run a standard PF on a modified system with constrained dynamics. The gist of this strategy is to restrict the exploration of the state space to typical solutions, while excluding others which are feasible but very unlikely. This approach can lead to algorithms which are numerically stable (i.e., they are more robust to the type of collapse described in \cite{Snyder08}) and yield acceptable accuracy with high probability. A classical method to constrain the filter dynamics is to truncate the support of the state $\X(t)$, a solution that has been explored algorithmically for Kalman filters \cite{Garcia12,Amor17} and also exploited to analyse filter stability \cite{Crisan20} and the convergence of PFs \cite{Heine08}.

%
\subsection{Contributions}

We introduce a class of constrained PFs for discretely-observed It\^o SDEs. Compared to earlier schemes in \cite{Heine08} or \cite{Crisan20}, the constraint is applied to the system dynamics (i.e., the SDE) instead of the likelihood function. It restricts the support of the state $\X(t_n)$ to a compact subset of the state space at each observation time $t_n$. 

Assuming an exact implementation of the constraint, we present a convergence analysis of the resulting PF. We analytically obtain error bounds that account both for the discrepancy with respect to (w.r.t.) the optimal, unconstrained Bayesian filter and for the Monte Carlo sampling errors. Under standard assumptions (following \cite{DelMoral04} or \cite{Kunsch05}), these bounds hold uniformly over time, even if $M,T \to \infty$. Based on the framework of \cite{Akyildiz24}, we also extend our analysis to account for the effect of time-discretisation (i.e., the numerical scheme used for the simulation of the underlying SDE) and obtain error bounds that incorporate these errors as well. The latter analysis is carried out under smoothness assumptions and yields convergence of the constrained PF for any finite time horizon $T<\infty$.

Finally, we carry out a numerical study to illustrate the practical application of the methodology. The model of choice for the computer simulations is a stochastic Lorenz 96 system with partial and noisy observations. The constraint on the PF is implemented via a modification of the drift in the underlying SDE using `barrier functions' as proposed in \cite{Erdogan26}. This implementation is only approximate, but guarantees the restriction of the state space and it is computationally much simpler than exact implementations based on rejection sampling or the Doob $h$-transform \cite{Erdogan26,Pieper-Sethmacher25b}.

%
\subsection{Organization of the paper}

The continuous-discrete filtering problem and the standard PF for this setting are described in Section \ref{sBackground}. In Section \ref{sConstrained} we introduce the new class of constrained PFs for discretely-observed It\^o diffusions and analyse their convergence under different sets of regularity assumptions. A numerical example is presented in Section \ref{sNumerics}. Technical proofs are detailed in Section \ref{sProofs} and Section \ref{sConclusions} is devoted to the conclusions.

%
\section{Filtering in discretely-observed It\^o diffusions} \label{sBackground}

%
\subsection{State space models} \label{ssIntroModels}

The signal of interest is modelled as a continuous-time It\^o process $\X(t)$ defined on a probability space $(\Omega,\mF,\mbP)$ and described by the stochastic differential equation (SDE) 
\beq\label{eq:sde}
\sd \X(t) = a(\X(t),t) \sd t + s(\X(t),t) \sd \W(t),  
\quad t \in [0,T],
\eeq
where 
\begin{itemize}
\item $t$ denotes continuous time and $T$ is some time horizon, 
\item $a:\mX \times [0,T] \mapsto \mX$ is the drift function that maps the signal space $\mX \subseteq \mbR^{d_x}$ into itself at any time $t$,
\item $\W(t)$ is a $d_w$-dimensional Wiener process, 
\item and $s:\mX \times [0,T] \mapsto \Real^{\dx \times d_w}$ is the diffusion (matrix) coefficient.
\end{itemize} 
The $i$-th entry of the $d_X$-dimensional drift term is denoted $a_i(\X(t),t)$; similarly, the entry in the $i$-th row and $j$-th column of the diffusion coefficient is denoted $s_{i,j}(\X(t),t)$. We construct a discrete-time sequence $X_n := \X(t_n)$ by sampling the process $\X(t)$ on the grid
\beq
0 = t_0 < t_1 < \cdots < t_n < \cdots < t_M = T.
\label{eqGrid}
\eeq
The dynamics of the random sequence $X_n$ are governed by Markov kernels $K_n: \mX \times \mB(\mX) \mapsto [0,1]$, where $n=1, \ldots, M$, and $\mB(\mX)$ is the Borel sigma-algebra of subsets of the state space $\mX$. Specifically, we define
\beq
K_n(x',\sd x) := \mbP\left( \X(t_n) \in \sd x | \X(t_{n-1})=x' \right), \quad \text{for $n \ge 1$}.
\nn
\eeq


The time grid \eqref{eqGrid} corresponds to the time instants when the signal is measured. We assume observations of the form 
\beq 
\label{eq: ObsModel}
Y_n = m(X_n) + U_n, \quad n=1, \ldots, M,
\eeq
where $U_n$ is a sequence of zero-mean, independent identically distributed (i.i.d.), $\dy$-dimensional real random variables (r.vs.), and the observation function $m: \mX \mapsto \mY$ maps the state space $\mX \subseteq \Real^{\dx}$ into the observation space $\mY \subseteq \Real^{\dy}$. The state $X_n$ and the observations $Y_n$ are related through a likelihood $g_n(x)$ that depends on the probability law of the noise variables $U_n$ and the observed value $Y_n= y_n$. For example, if the noise is zero-mean Gaussian with covariance matrix $\Sigma_u$, denoted $U_n \sim \mN(0,\Sigma_u)$, then
\beq
g_n(x) \propto \exp\left\{ 
    -\frac{1}{2}\left\|
        y_n - m(x)
    \right\|_{\Sigma_u^{-1}}^2
\right\}
\nn
\eeq
where $\left\| y_n - m(x) \right\|_{\Sigma_u^{-1}}^2 = \left(
    y_n - m(x)
\right)^\top \Sigma_u^{-1} \left(
    y_n - m(x)
\right)
$
and ${}^\top$ denotes transposition.

%
\subsection{Bayesian filtering}

Let $K=\{ K_n \}_{n \ge 1}$ be the sequence of Markov kernels, let $g=\{g_n\}_{n\ge 1}$ be the sequence of likelihoods for the grid in \eqref{eqGrid} and assume that $X_0=\X(t_0)$ has known probability law $\pi_0$. Then, we denote the system of interest as $\sS:=(\pi_0, K, g)$ and we refer to $\sS$ as a Markov state space model (SSM). 

We are interested in the marginal probability measures 
\beq
\xi_n(\sd x) := \mbP( X_n \in \sd x | Y_{1:n-1} = y_{1:n-1} )
~~\text{and}~~
\pi_n(\sd x) := \mbP( X_n \in \sd x | Y_{1:n} = y_{1:n} )
\nn
\eeq
generated by model $\sS$ for a fixed sequence of observations $\{ Y_1=y_1, \ldots, Y_n=y_n\}$. The measure $\pi_n$ is the marginal probability law of $X_n$ conditional on the available data up to time $t_n$, and it is often termed the optimal Bayesian filter. The measure $\xi_n$ is the marginal probability law of $X_n$ conditional on the data up to time $t_{n-1}$ and we refer to it as the (one step ahead) predictive distribution.

If we let $\mP(\mX)$ be the family of probability measures on $(\mX,\mB(\mX))$, then for any $\mu \in \mP(\mX)$ and any real measurable function $f : \mX \mapsto \Real$ we introduce the shorthand
\beq
\mu(f) := \int_\mX f\sd \mu.
\nn
\eeq
With this notation, we define the sequence of operators
$
\Phi_n : \mP(\mX) \mapsto \mP(\mX),
$
\beq
\Phi_n(\mu)(f) := \frac{
	K_n\mu(fg_n)
}{
	K_n\mu(g_n)
},
\label{eqPUop}
\eeq
where $K_n\mu(\sd z) := \int_{\mX} K_n(x,\sd z)\mu(\sd x)$. It is easy to show \cite{Crisan20} that the operators $\{\Phi_n\}_{n\ge 1}$ yield the sequence of probability measures
\beq
\pi_n := \Phi_n(\pi_{n-1}), \quad n\ge 1,
\eeq
generated by the model $\sS$. The composition of consecutive operators is denoted
\beq 
\Phi_{m:n} = \Phi_n \circ \Phi_{n-1} \circ \cdots \circ \Phi_{m+1},
\nn
\eeq
for $m\le n$, in such a way that $\pi_n = \Phi_{m:n}(\pi_m)$. We adopt the convention $\Phi_{n:n}(\mu)=\mu$. Also note that $\Phi_{n-1:n} = \Phi_n$. 

We often refer to $\Phi_n$ as a prediction-update (PU) operator, because it can be interpreted as comprising the one-step-ahead prediction $\xi_n = K_n\pi_{n-1}$ followed by the Bayesian update $\pi_n(\sd x) = \frac{g_n(x) \xi_n(\sd x)}{\xi_n(g_n)}$.

%
\subsection{Stability of the optimal filter}

Let $\dtv{\cdot}$ denote the total variation norm. The sequence of optimal filters $\pi_n=\Phi_n(\pi_{n-1})$, $n \ge 1$, is called {\em stable} when
$$
\lim_{n\to\infty} \dtv{ \Phi_{0:n}(\pi_0) - \Phi_{0:n}(\eta_0) } = 0
$$
no matter the choice of a priori distributions $\pi_0,\mu_0 \in \mP(\mX)$. It is apparent that that the stability of $\pi_n$ is a property of the operators $\Phi_n$ (see \cite{Crisan20} for further details). A sufficient condition for stability is that the Markov kernels in $K$ are {\em mixing} and the likelihood functions in $g$ integrable w.r.t. some reference probability density functions (pdfs) $u=\{u_n\}_{n\ge 1}$ on $\mX$ \cite{DelMoral01c,DelMoral04,Kunsch05}.

\begin{definition} \label{defMixing}
The Markov kernels $K=\{K_n\}_{n\ge 1}$ are mixing with constant $0<\gamma<1$ when
\begin{enumerate}[(a)]
\item $K_n(x',\sd x) = k_n(x',x) \sd x$ for every $n \ge 1$, where $k_n: \mX \times \mX \mapsto [0,\infty)$ is a pdf w.r.t. the Lebesgue measure, and
\item there is a sequence of pdfs $u_n:\mX\mapsto [0,\infty)$ such that
\beq
\gamma u_n(x) \le k_n(x',x) \le \frac{1}{\gamma} u_n(x)
\quad \text{for every $n \ge 1$ and every $x' \in \mX$.}
\nn
\eeq
\end{enumerate}
\end{definition}

\begin{lemma} \label{lmStability}
Assume that the Markov kernels $K=\{ K_n \}_{n \ge 1}$ are mixing with constant $\gamma>0$ and pdfs $u=\{u_n\}_{n\ge 1}$, and the likelihoods $g=\{g_n\}_{n\ge 1}$ are integrable, namely,
\beq
\int_\mX g_n(x)u_n(x) \sd x < \infty, \quad \text{for all $n\ge 1$}. 
\nn
\eeq
Then, for any $n \ge m \ge 0$,
\beq
\dtv{\Phi_{m:n}(\mu) - \Phi_{m:n}(\eta)} \le \frac{(1-\gamma^2)^{n-m}}{\gamma^2} \dtv{\mu - \eta}
\nn
\eeq
and, in particular, $\lim_{n\to\infty} \dtv{ \Phi_{0:n}(\mu) - \Phi_{0:n}(\eta) } = 0$ for any $\mu,\eta \in \mP(\mX)$.
\end{lemma}

This is a well known result (see Lemma 10 in \cite{Kunsch05}). The mixing condition in Definition \ref{defMixing} is sufficient but not necessary for stability. 

%
\subsection{Particle filters}

PFs are recursive Monte Carlo algorithms used to approximate the probability laws $\{\pi_n\}_{n\ge 0}$. The simplest scheme is the bootstrap filter or standard PF \cite{Gordon93,Kitagawa96} (see also \cite{Liu98,Doucet00,Djuric03}). It is sketched as Algorithm \ref{alPF} below. The Monte Carlo estimator of $\pi_n$ is 
\beq
\pi_n^N(\sd x) = \frac{1}{N}\sum_{i=1}^N \delta_{x_n^i}(\sd x),
\eeq 
where $\delta_{x}$ denotes the Dirac delta measure located at $x \in \mX_n$ and $x_n^1, \ldots, x_n^N$ are Monte Carlo samples on the state space $\mX$. Many more sophisticated versions of this algorithm can be devised, including the use of proposal distributions different from the Markov kernel $K_n(x,\sd x)$ and more efficient resampling schemes. See \cite{Doucet00,Djuric03,DelMoral04,Cappe07,Bain08,Sarkka13,Chopin20} and references therein. 

\begin{algorithm}
\caption{Standard PF for a SSM $\sS=(\pi_0,K,g)$.}
\label{alPF}
\begin{itemize}
\item[] \textit{Inputs}: model $\sS=(\pi_0,K,g)$ and number of particles $N$.\\
\textit{Outputs}: approximation $\pi_n^N$ of the optimal filter $\pi_n$ at every step $n\ge 1$.
\item \textbf{Initialisation.} At time $n=0$, draw $N$ iid samples $x_0^i \sim \pi_0$, $i=1, \ldots, N$.
\item \textbf{Recursive step.} At time $n\ge 1$:
	\begin{enumerate}
	\item For $i=1, \ldots, N$, 
		\begin{itemize}
		\item simulate $\bar x_n^i \sim K_n(x_{n-1},\sd x)$ and
		\item compute $\bar w_n^i = g_n\left(\bar x_n^i\right)$.
		\end{itemize}
	\item Normalise the weights, $w_n^i = \frac{\bar w_n^i}{\sum_{j=1}^N \bar w_n^j}$, $i=1, \ldots, N$.
	\item Resample: let $\bar \pi_n^N(\sd x) = \sum_{j=1}^N w_n^j \delta_{\bar x_n^j}(\sd x)$ and draw $N$ iid samples $x_n^i \sim \bar \pi_n^N$, $i=1, \ldots, N$. 
	\end{enumerate}
\end{itemize}
\end{algorithm}

Consistency of the standard PF can be proved under very mild assumptions \cite{Crisan00,DelMoral00,Bain08,Crisan14a} even when the observations are not conditionally independent or the state $X_n$ is not Markov \cite{Miguez13b}. If the sequence of optimal filters $\pi_n$, $n \ge 1$, is stable then (under some additional regularity conditions) $\pi_n^N$ converges uniformly over time towards the optimal filter $\pi_n$ \cite{DelMoral01c,DelMoral04,Kunsch05,Whiteley13b}. In particular, for any bounded and measurable test function $f:\mX\mapsto\mbR$,
\beq
\sup_{n \ge 1} \| \pi_n(f) - \pi_n^N(f) \|_p \le \frac{C_p \|f\|_\infty}{\sqrt{N}}
\nn
\eeq
for all $p \ge 1$, where $C_p$ is a finite constant independent of $N$ and time $n$.

%
\section{Particle filters on constrained state space models} \label{sConstrained}

%
\subsection{Constrained models} \label{ssConstrainedModels}

Let $\sS = (\pi_0, K, g)$ be a Markov SSM and let $\mC=\{\mC_n\}_{n \ge 0}$ be a sequence of compact subsets of the state space $\mX$. We construct a constrained SSM $\hat \sS = (\hat \pi_0,\hat K,g)$ by modifying the prior distribution and the family of kernels, namely,
\beq
\hat K_n(x',\sd x) := \frac{
    \Ind_{\mC_n}(x) K_n(x',\sd x)
}{
    \int_{\mC_n} K_n(x',\sd x)
} 
\quad \text{and} \quad
\hat\pi_0(\sd x) = \frac{\Ind_{\mC_0}(x)\pi_0(\sd x)}{\int_{\mC_0} \pi_0(\sd x)},
\label{eqBasicConstraint}
\eeq
where 
$$
\Ind_{\mC_n}(x) = \left\{ 
    \begin{array}{ll} 
    1, &\text{if $x\in\mC_n$}\\
    0, &\text{otherwise}\\
    \end{array}
\right.
$$
is the indicator function. It is apparent that $\hat K_n(x',\mX)=1$ for any $x'\in\mX$ whenever $K(x',\mX)=1$, hence the kernel $\hat K_n$ is well defined if the original kernel $K$ is well defined.

The PU operators for $\hat\sS$ are constructed as in Eq. \eqref{eqPUop}, i.e.,
\beq
\hat\Phi_n(\mu)(f) := \frac{
    \hat K_n \mu(fg_n)
}{
    \hat K_n\mu(g_n)
}
\label{eqCPUop}
\eeq
and, given an initial distribution $\hat\pi_0$, they yield the sequence of marginal laws $\hat\pi_n = \hat\Phi_n(\hat\pi_{n-1}) = \hat \Phi_{0:n}(\hat\pi_0)$, where $\hat\Phi_{0:n}=\hat\Phi_n \circ \cdots \circ \hat\Phi_1$. Moreover, 
\beq
\hat\pi_n(\sd x) = \hat\Phi_n(\hat\pi_{n-1})(\sd x)= \frac{
    g_n(x)\hat\xi_n(\sd x)
}{
    \hat\xi_n(g_n)
},
\nn
\eeq
where $\hat\xi_n(\sd x) = \hat K_n\hat\pi_{n-1}(\sd x)$ is the constrained one-step-ahead predictive law.

%
\subsection{Stability}

We assume some regularity for the (unconstrained) SSM of interest $\sS$.

\begin{assumption} \label{assModel1}
The elements of model $\sS=(\pi_0,K,g)$ satisfy the conditions below:  
\begin{itemize}
\item there are density functions $k_n:\mX\times\mX\mapsto(0,\infty)$ such that $K_n(x',\sd x)=k_n(x',x)\sd x$ (in particular, $k_n(x',x)>0$ for any $x',x\in \mX$ and $n\ge 1$),
\item $\|k\|_\infty := \sup_{n,x',x} k_n(x',x) < \infty$, and 
\item $g_n>0$ for every $n\ge 1$ and $\|g\|_\infty := \sup_{n\ge 1} \|g_n\|_\infty < \infty$.
\end{itemize}
\end{assumption}

The sufficient conditions for stability in Lemma \ref{lmStability} hold in a natural way for a constrained SSM $\hat \sS=(\hat \pi_0, \hat K, g)$ constructed from a model $\sS=(\pi_0,K,g)$ that satisfies Assumption \ref{assModel1}. This is made precise by Theorem \ref{thStabilitySl} below, which is proved in \cite{Erdogan26} (see Theorem 3.2 therein). 

\begin{theorem} \label{thStabilitySl}
Let $\sS=(\pi_0,K,g) \in \mbS$ be a SSM for which Assumption \ref{assModel1} holds. If we choose a sequence $\mC=\{\mC_n\}_{n\ge 0}$ such that 
\beq
\inf_{n\ge 1} \inf_{(x',x) \in \mC_{n-1} \times \mC_n} k_n(x',x) \ge \bar\gamma
\label{eqlowerboundinC}
\eeq
for some $\bar\gamma>0$, then 
\beq
\dtv{\hat \Phi_{0:n}(\mu) - \hat\Phi_{0:n}(\eta)} \le \frac{(1-\gamma^2)^n}{\gamma^2}\dtv{\mu-\eta},
\label{eqStabHatS}
\eeq
where $\hat\Phi_n$ is the constrained PU operator in \eqref{eqCPUop} and $\gamma=\frac{\bar\gamma^2}{\|k\|_\infty^2>0}$. In particular, 
$$
\lim_{n\to\infty} \dtv{\hat \Phi_{0:n}(\mu) - \hat\Phi_{0:n}(\eta)} = 0 
\quad \text{for any $\mu,\eta \in \mP(\mX)$,}
$$ 
hence the SSM $\hat \sS$ generates a stable sequence of marginal probability laws $\{\hat \pi_n\}_{n\ge 0}$.
\end{theorem}

\begin{remark}
While Gaussian kernels are not mixing according to Definition \ref{defMixing}, they satisfy the conditions in Assumption \ref{assModel1} and, hence, can be constrained to yield a stable sequence $\hat \pi_n$, $n \ge 0$.
\end{remark}

%
\subsection{Approximation errors} \label{ssApproxErr}

We seek constrained models $\hat \sS$ for which the posterior marginals $\hat \pi_n$ are close to the marginals $\pi_n$ generated by the original model $\sS$ under (reasonable) regularity assumptions. To this end, we adopt the framework of \cite{Crisan20}. Specifically, let $\{\mC^l\}_{l\ge 1}$ be a sequence where each element $\mC^l = \{ \mC_n^l \}_{n \ge 0}$ is a family of compact subsets of the state space $\mX$ that can be used to constrain the model $\sS$. In this way, we can construct a sequence of constrained models $\hat \sS^l=(\hat \pi_0^l,\hat K^l, g)$, $l \ge 1$, where $\hat K^l = \{\hat K_n^l\}_{n\ge 1}$, 
$$
\hat K_n^l(x',\sd x) = \frac{
    \Ind_{\mC_n^l}(x)K_n(x',\sd x)
}{
    \int_{\mC_n^l} K_n(x',\sd x)
}
\quad
\text{and}
\quad
\hat \pi_0^l(\sd x) = \frac{\Ind_{\mC_0^l}(x)\pi_0(\sd x)}{\int_{\mC_0^l} \pi_0(\sd x)}.
$$
If $\mC^l$, $l \ge 1$ are chosen in such a way that $\lim_{l\to\infty} \mC_n^l = \mX$ for every $n$, then $\lim_{l\to\infty} \hat \sS^l = \sS$ and, by Lemma 2.4 in \cite{Crisan20},
\beq
\lim_{l\to\infty} \hat \pi_n^l = \lim_{l\to\infty} \hat \Phi_{0:n}^l(\hat \pi_0^l) = \pi_n \quad \text{in total variation,}
\nn
\eeq
where the operators $\hat\Phi_n^l$, $n \ge 1$, are constructed in the obvious way given the kernels $\hat K_n^l$ and Eq. \eqref{eqCPUop}. We also denote $\hat\xi_n^l(\sd x)=\hat K_n^l\hat \pi_{n-1}^l(\sd x)$.

An arbitrary constraint $\mC^l=\{\mC_n^l\}_{n\ge 0}$ on $\mX$ may not lead to a ``reasonable'' approximation of the SSM $\sS$. In general, the subsets $\mC_{n}^l$ should be selected to capture a significant mass of the transition probability, i.e., in such a way that, given $x'\in\mC_{n-1}^l$, the transition probability 
\beq
\mbP(X_n \in \mC_n^l| X_{n-1}=x') = \int_{\mC_n^l} K_n(x',\sd x)
\nn
\eeq
can be ensured to be large enough. Intuitively, the sequence of subsets $\{\mC_n^l\}_{n\ge 1}$ should be ``in agreement'' with the (unconstrained) dynamics of the state $X_n$. Hereafter, we abide by Assumption \ref{assCl-2} below, which imposes a relatively strong form of such ``agreement'' between the dynamics of $X_n$ and the subsets $\mC_n^l$. See Section 3 in \cite{Erdogan26} for a more detailed analysis, including weaker assumptions on $\mC^l$.  

\begin{assumption} \label{assCl-2}
The family of compact subsets $\{\mC_n^l\}_{n\ge 0, l\ge 1}$ of the state space $\mX$ is constructed to guarantee that, for every $n\ge 0$, $\lim_{l\to\infty} \mC_n^l = \mX$ and there is a single increasing sequence $\varepsilon^l > 0$ such that
\beq
\pi_0(\mC_0^l) \ge \varepsilon^l, \quad \inf_{x'\in\mC_{n-1}^l} \int_{\mC_n^l} K_n(x',\sd x) \ge \varepsilon^l
\quad \text{and} \quad
\lim_{l\to\infty} \varepsilon^l = 1 ~~\forall n \ge 0.
\nn
\eeq
\end{assumption}

Finally, we assume that the sequence $\pi_n$, $n \ge 1$, generated by the original model $\sS$ is stable. This is done in order to find a bound on the approximation error $| \pi_n(f) - \hat\pi_n^l(f) |$ that holds uniformly over time $n$. We do not assume that kernels $K_n$ to be necessarily mixing, though. A slightly weaker assumption (stability with a possibly non-exponential rate) is sufficient for our purpose.

\begin{assumption} \label{assModel2}
The model $\sS=(\pi_0,K,g) \in \mbS$ yields a stable sequence of optimal filters. In particular, there is a decreasing sequence $\{r_i\}_{i\ge 0}$ such that $r_0=1$, $\sum_{i=0}^\infty r_i < \infty$ and
$$
\dtv{\Phi_{m:n}(\mu)-\Phi_{m:n}(\eta)} \le r_{n-m}\dtv{\mu-\eta}
$$
for any pair of probability measures $\mu,\eta \in \mP(\mX)$.
\end{assumption}

If the original model $\sS$ satisfies the stability Assumption \ref{assModel2} and the constraints $\mC^l$, $l\ge 1$, satisfy Assumption \ref{assCl-2} then the approximation error $| \pi_n(f) - \hat\pi_n^l(f) |$ vanishes, uniformly over discrete time $n$, as $l \to \infty$. Specifically, we have the following result, proved in \cite{Erdogan26} (see Theorem 3.8). 

\begin{theorem} \label{thUniformDTV-1}
Let $\sS=(\pi_0,K,g)$ be a SSM for which Assumption \ref{assModel2} holds. Let $\{\mC^l\}_{l\ge 1}$ be a sequence of constraints of the state space $\mX$ that satisfies Assumption \ref{assCl-2}, and let $\hat \sS^l=(\hat \pi_0^l,K^l,g)$, $l \ge 1$, be the resulting sequence of constrained SSMs. If there are constants $\zeta>0$ and $\| g \|_\infty < \infty$ such that $\inf_{n\ge 1, x\in\mX} g_n(x) > \zeta$ and $\sup_{n \ge 1, x\in \mX} g_n(x) \le \| g \|_\infty$, respectively, then there is a finite constant $c<\infty$, independent of $n$, such that
\beq
\sup_{|f|\le 1} \left| \pi_n(f) - \hat \pi_n^l(f) \right| \le c \frac{1-\varepsilon^l}{\varepsilon^l}.
\label{eqThUnifDTV-1}
\eeq
In particular, $\lim_{l\to\infty} \sup_{n\ge 0} \dtv{\pi_n-\hat\pi_n^l}=0$.
\end{theorem}

The assumption $\inf_{n\ge 1, x\in\mX} g_n(x) > \zeta > 0$ is strong, but common in analyses aimed at proving uniform convergence of PFs \cite{DelMoral04,Kunsch05}. If the observations have the form $Y_n=m(X_n)+U_n$, where $m(\cdot)$ is an observation function and $U_n$ is Gaussian noise, then a bounded function $m(\cdot)$ with bounded observations, $\sup_n \| y_n \| < \infty$, yields $\inf_n g_n > \zeta >0$.

%
\subsection{Constrained particle filters} \label{ssCPFs}

Given a model $\sS^l = (\hat \pi_0^l,K^l,g)$ one can apply a standard PF (Algorithm \ref{alPF}) to obtain a particle approximation $\hat \pi_n^{l,N} = \frac{1}{N}\sum_{i=1}^N \delta_{x_n^i}$ of the constrained marginal law $\hat\pi_n^l = \hat\Phi_{0:n}^l(\hat \pi_0^l)$. In this section we investigate whether $\hat \pi_n^{l,N}$ can also be used as an approximation of the optimal (unconstrained) Bayesian filter $\pi_n$. 

We adopt the assumption $\inf_{n\ge 1} g_n>\zeta>0$ in order to obtain uniform-over-time error bounds comparable to the results in \cite{DelMoral01c,DelMoral04,Kunsch05}. Moreover, since the particle estimator $\hat \pi_n^{l,N}$ is a random probability measure, the approximation error $| \pi_n(f) - \hat\pi_n^{l,N}(f)|$ is a r.v.. Hence, we provide bounds for the $L^p$ norms of this error, namely 
\beq
\| \pi_n(f) - \hat\pi_n^{l,N}(f)\|_p = \mbE\left[ | \pi_n(f) - \hat\pi_n^{l,N}(f)|^p \right]^{\frac{1}{p}}
\quad \text{for $p \ge 1$.}
\nn
\eeq

\begin{theorem}\label{thCPFs}
Let $\sS=(\pi_0,K,g)$ be a SSM for which Assumptions \ref{assModel1} and \ref{assModel2} hold and, in addition, there is a constant $\zeta>0$ such that $\inf_{n\ge 1} g_n > \zeta$. Let $\{\mC^l\}_{l\ge 1}$ be a sequence of constraints of the state space $\mX$ that satisfies Assumption \ref{assCl-2} and the inequality
\beq
\inf_{n\ge 1, (x',x)\in\mC_{n-1}^l\times\mC_n^l} k_n(x',x) \ge \bar \gamma_l > 0.
\label{eqThCPFsAsC}
\eeq
Then, for any $p \ge 1$ and $\gamma_l := \frac{\bar\gamma_l^2}{\|k\|_\infty^{2}}$,
\beq
\sup_{n \ge 1, |f|\le 1} \| \pi_n(f) - \hat \pi_n^{l,N}(f) \|_p < c\frac{1-\varepsilon^l}{\varepsilon^l}  + \frac{c_p\|g\|_\infty}{\zeta\gamma_l^3 \sqrt{N}},
\label{eqthCPF1}
\eeq
where $c_p<\infty$ is a constant that depends only on $p$ and $c<\infty$ is the constant in \eqref{eqThUnifDTV-1}. In particular, for every $l<\infty$,
\beq
\lim_{N\to\infty} \sup_{n \ge 1, |f|\le 1} \| \pi_n(f) - \hat \pi_n^{l,N}(f) \|_p < c\frac{1-\varepsilon^l}{\varepsilon^l}
\label{eqthCPF2}
\eeq
and there is a constant $\bar c <\infty$, independent of $N$ and $n$, such that
\beq
\sup_{n\ge 0, |f|\le 1} \lim_{l\to\infty} \| \pi_n(f) - \hat \pi_n^{l,N}(f) \|_p < \frac{\bar c}{\sqrt{N}}.
\label{eqthCPF3}
\eeq
\end{theorem}

See Sections \ref{ssAncillary} and \ref{ssCPFs} for a complete proof. Theorem \ref{thCPFs} holds even if $T\to\infty$ and $M\to\infty$. From inequalities \eqref{eqThUnifDTV-1} and \eqref{eqthCPF2} we see that the approximation error $\| \pi_n(f) - \hat \pi_n^{l,N}(f) \|_p$ can be made comparable to the error $| \pi_n(f) - \hat\pi_n^l(f) |$ when $N$ is sufficiently large. In this sense, the constrained PF yields an approximation of the unconstrained optimal filter $\pi_n$.

%
\subsection{Discretisation errors} \label{ssDiscretisation}

The implementation of a standard PF for a model $\sS=(\pi_0,K,g)$ relies on the assumption that one can at least generate samples from the distributions $K_n(x',\sd x)$ exactly. However, the kernels induced by a nonlinear SDE can only be simulated using numerical procedures \cite{Kloeden95,Kloeden12}. To be specific, let us assume that we implement a numerical scheme to draw approximate samples from $K_n(x',\sd x)$. Such numerical scheme implicitly defines an approximate kernel that we denote as $K_n^h(x', \sd x)$, where $h>0$ indicates the time-step of the simulation. For example, a simple Euler-Maruyama scheme yields
\beq
\tilde X_{n-1,0} = x', \quad \text{and} \quad
\tilde X_{n-1,j} = \tilde X_{n-1,j-1} + h~a(\tilde X_{n-1,j-1} ) + \sqrt{h} s(\tilde X_{n-1,j-1}) Z_j, \nn
\eeq
for $j = 1, \ldots, J$, where $J=(t_n-t_{n-1})/h$ and $Z_j$ is a sequence of i.i.d. $\mN(0,I)$-distributed r.v.s. The variate $X_n^h := \tilde X_{n-1,J}$ is a draw from the approximate kernel $K_n^h(x',\sd x)$.

In most practical problems, we can only work with the approximate SSM $\sS^h = (\pi_0, K^h, g)$, where $K^h = \{ K_n^h\}_{n \ge 1}$, rather than the {\em true} model $\sS=(\pi_0,K,g)$. As a consequence, we also have approximations of the optimal filter and the predictive distributions, denoted $\pi_n^h$ and $\xi_n^h$, respectively. The PU operator for model $\sS^h$ has the form
\beq
\Phi_n^h(\mu)(f) = \frac{K_n^h\mu(fg_n)}{K_n^h\mu(g_n)}, 
~~ \text{hence} ~~
\pi_n^h = \Phi_n^h(\pi_{n-1}^h) ~~\text{and}~~ \pi_n^h=\Phi_{m:n}(\pi_m^h)
\nn
\eeq
with obvious notation. Note that $\pi_0^h = \pi_0$.

A natural question is whether we can quantify the discrepancy between the optimal filter $\pi_n$ and its approximation $\pi_n^h$. 
To address this problem, let us assume that the numerical scheme used to simulate the SDE \eqref{eq:sde} converges weakly \cite{Kloeden95}, with order at least 1. This is made precise by Assumption \ref{asW} below.

\begin{assumption} \label{asW}
Let $C_B^4(\mbR^{d_x})$ be the space of real, continuous and bounded functions with uniformly bounded derivatives up to order 4. For any $v\in C_B^4(\mbR^{d_x})$ and any $x'\in\mX$,
$$
\left|
    \int v(x)K_n(x',\sd x) - \int v(x)K_n^h(x',\sd x)
\right| \le \bar C_n h, ~~n = 1, \ldots, M,
$$
where every $\bar C_n$ is a finite constant independent of $x'$ and $h$.
\end{assumption}

\begin{remark}
Weak convergence of the simulation scheme that yields the approximate kernel $K_n^h$ typically requires some regularity of both the drift $a(\cdot,t)$ function and the diffusion $s(\cdot,t)$ coefficient, e.g., sufficient smoothness and (at most) polynomial growth \cite{Kloeden95}.
\end{remark}

\begin{theorem} \label{thMarginals}
Let $Y_{1:M} = y_{1:M}$ be an arbitrary but fixed sequence of observations, let Assumption \ref{asW} hold and choose a test function $f \in C_B^4(\Real^{d_x})$ with $\|f\|_\infty \le 1$. If, for every $n=1, \ldots, M$,
\begin{itemize}
\item $g_n>0$, $\| g_n \|_\infty \le 1$ and $g_n \in C_B^4(\Real^{d_x})$, and 
\item $\bar f_n \in C_B^4(\Real^{d_x})$, where
$
\bar f_n(x') = \int f(x) K_n(x',\sd x),
$
\end{itemize}
then there are finite constants $\{ \tilde C_n, C_n \}_{n=1}^M$, independent of $h$, such that
\beq
\left|
    \xi_n(f) - \xi_n^h(f)
\right| \le \tilde C_n h 
\quad \text{and} \quad 
\left|
    \pi_n(f) - \pi_n^h(f)
\right| \le C_n h \quad \text{for $n=1, \ldots, M$}.
\nn
\eeq
\end{theorem}

See Section \ref{ssThDiscretisation} for a proof.

We can build on Theorem \ref{thMarginals} to compute bounds for the combined errors due to time discretisation, the constraint $\mC^l = \{\mC_n^l\}_{n\ge 0}$ and the particle approximation generated by the standard PF. Let $\sS=(\pi_0,K,g)$ be the initial model, let $\sS^h=(\pi_0,K^h,g)$ be the model with a discretised kernel and let $\sS^{h,l} = (\hat\pi_0^l,K^{h,l},g)$ be the constrained version of $\sS^h$. Model $\sS^{h,l}$ yields the marginals $\hat\pi_n^{h,l} = \hat\Phi_{0:n}^{h,l}(\hat\pi_0^l)$. Assume we run Algorithm \ref{alPF} on model $\sS^{h,l}$ to obtain the Monte Carlo approximation $\hat\pi_n^{h,l,N} = \frac{1}{N}\sum_{i=1}^N \delta_{x_n^i}$. The following result is rather straightforward.

\begin{corollary} \label{corDiscreteTruncate}
Let $Y_{1:M} = y_{1:M}$ be an arbitrary but fixed sequence of observations, let Assumption \ref{asW} and Assumption \ref{assCl-2} hold for the models $\sS$ and $\left\{ \sS^{h,l} \right\}_{l\ge 1}$, and choose a test function $f \in C_B^4(\Real^{d_x})$ with $\| f \|_\infty \le 1$. If, for every $n=1, \ldots, M$,
\begin{itemize}
\item $g_n>0$, $\| g_n \|_\infty \le 1$ and $g_n \in C_B^4(\Real^{d_x})$, and 
\item $\bar f_n \in C_B^4(\Real^{d_x})$, where
$
\bar f_n(x') = \int f(x) K_n(x',\sd x),
$
\item and we choose $h \le N^{-\frac{1}{2}}$,
\end{itemize}
then there are finite constants $\{ \tilde c_n, \bar c_n \}_{n=1}^M$, independent of $N$ and $h$, such that
\beq
\| \pi_n(f) - \hat \pi_n^{h,l,N}(f) \|_p \le \frac{
	\tilde c_n
}{
	\sqrt{N}
} + \bar c_n \frac{1-\varepsilon^l}{\varepsilon^l}, ~~\text{for $n=1, \ldots, M$}.
\nn
\eeq
 
\end{corollary}

\begin{proof}
An iterated triangle inequality yields
\beq
\| \pi_n(f) - \hat \pi_n^{h,l,N}(f) \|_p \le
| \pi_n(f) - \pi_n^h(f) | +
| \pi_n^h(f) - \hat \pi_n^{h,l}(f) | + 
\| \hat \pi_n^{h,l}(f) - \hat \pi_n^{h,l,N}(f) \|_p. 
\nn
\eeq
From Theorem \ref{thMarginals} we obtain 
\beq
| \pi_n(f) - \pi_n^h(f) | \le C_n h \le \frac{C_n}{\sqrt{N}};
\label{eqITI1}
\eeq
from Corollary 3.6 in \cite{Erdogan26} we have
\beq
| \pi_n^h(f) - \hat \pi_n^{h,l}(f) | \le \bar c_n \frac{1-\varepsilon^l}{\varepsilon^l}
\label{eqITI2}
\eeq
for some $\tilde c_n<\infty$, and Lemma 1 in \cite{Miguez13b} yields
\beq
\| \hat \pi_n^{h,l}(f) - \hat \pi_n^{h,l,N}(f) \|_p \le \frac{\breve c_n \|f\|_\infty}{\sqrt{N}}
\label{eqITI3}
\eeq
for some $\breve c_n<\infty$. Combining the inequalities \eqref{eqITI1}, \eqref{eqITI2} and \eqref{eqITI3} completes the proof, with constant $\tilde c_n = C_n + \breve c_n\|f\|_\infty$. \qed
\end{proof}

Note that stability of either $\pi_n$ or $\pi_n^h$ is {\em not} required for Corollary \ref{corDiscreteTruncate} to hold. 

%
\section{A numerical example} \label{sNumerics}

%
\subsection{Stochastic Lorenz 96 model} 

To illustrate numerically both the analysis in Section \ref{sConstrained} and the methodology described in Section \ref{sConstrained}, we consider the $\dx$-dimensional stochastic Lorenz 96 model with additive noise \cite{Perez-Vieites18,Grudzien20} described by the system of SDEs
\beq
\sd \X_i(t) = \left[ \X_{i-1}(t) \left(  \X_{i+1}(t) - \X_{i-2}(t)  \right) - \X_i(t) + F\right] dt + \sigma_x  \sd \W_i(t), \quad i=0, \ldots, \dx-1,
\label{eqL96}
\eeq
with periodic boundaries given by 
\beq
\X_{-2}(t) = \X_{\dx-2}(t),
\quad \X_{-1}(t) = \X_{\dx-1}(t),
\quad  \text{and} \quad  \X_0(t) = \X_{\dx}(t), 
\nn
\eeq
where the forcing constant is $F=8$ (which yields turbulent dynamics), $\W_i(t)$, $i=0, \ldots, \dx-1$, are standard Wiener processes and $\sigma_x = \sqrt{1/2}$ is a constant. The complete state is the $\dx \times 1$ vector $\X(t)=\left[ \X_0(t), \ldots, \X_{\dx-1}(t) \right]^\top$. 

We assume a linear and Gaussian observation model, namely 
\beq
Y_n = H^\top X_n + U_n,
\label{eqObsEq}
\eeq
where $U_n \sim \mN(0,\sigma_y^2 I_{d_y})$, $\sigma_y^2=1/2$ and $H$ is a $\dx \times \dy$ matrix and $X_n=\X(t_n)$. For each independent simulation trial, we generate $H$ randomly as
\beq
H = \left[
 	e_{m_{1}}, e_{m_{2}}, \ldots, e_{m_{\dy}} 
 \right] + V 
\nn
\eeq
where $e_{m}$ is a $\dx \times 1$ vector of 0s with a single value of 1 in the $m$-th entry, $V$ is a $\dx \times \dy$ matrix whose entries are independent Gaussian random variates, namely $V_{i,j} \sim \mN(0,\sigma_v^2)$ and $\sigma_v = 5 \times 10^{-4}$. The indices $m_1, \ldots, m_{d_y}$ are drawn randomly from the set $\{1, \ldots, \dx\}$, with uniform probabilities and no replacement. Intuitively, the $i$-th observation $Y_i$, $i\in\{1, \ldots, \dy\}$, corresponds to a linear combination of all the state variables, with $X_{m_i}$ being dominant and all other variables causing (individually small) interference in the measurement.
 
We approximate the ground-truth states $X_n$ in our experiments below by generating a discrete-time sequence $\X_l \approx \X(t_l')$, where $t_l'= l h$, by way of the standard Euler-Maruyama scheme with step size $h=10^{-3}$. We assume that observations are collected every $h_o=0.1$ continuous time units. Hence, $X_n = \X(nh) \approx \X_{\frac{h_o}{hn}} = \X_{100n}$. The Euler-Maruyama method is used in the numerical implementation of all the algorithms in this section.

%
\subsection{Numerical results} 

In this section we illustrate numerically the effect of constraining the state space on the approximation of the optimal filter by a standard PF. We do not attempt to evaluate the errors due to time discretisation (i.e., we ignore the discrepancy between $K_n$ and $K_n^h$ and assume that $K_n^h=K_n$). 

We consider a single constraint $\mC = \{\mC_n\}_{n=0}^M$, where $\mC_0$ is a hypercube with side length $3\sigma_x$ and every $\mC_n$ is a superlevel set for the log-likelihood $g_n$, namely,
$$
\mC_n = \{ x \in \mbR^{d_x} : \log g_n(x) > -8 \}.
$$
The constraint is implemented using a barrier function as described  in \cite{Erdogan26} (see Sections 4 and 5.2 therein for details).

Figure \ref{fNMSE} illustrates the performance of a standard PF, an auxiliary PF \cite{Pitt01} and a constrained standard PF implemented with a barrier function when the dimension of the Lorenz 96 model increases from $d_x = 50$ to $d_x = 400$ and the number of particles is kept fixed, $N=100$. The figure displays the normalised mean square error (MSE), computed as 
$$
NMSE = \frac{\sum_{n=1}^M \| X_n - \hat X_n \|^2}{\sum_{n=1}^M \| X_n \|^2}, 
$$
where $X_n$ is the true state and $\hat X_n$ is its estimate (i.e., NMSE above is the power of the approximation error normalised by the power of the state signal). The values in the plot are the average of 30 independent simulation runs and the shaded areas correspond to one standard deviation. This experiment shows that the PFs running on the unconstrained model fail to track the Lorenz 96 system even for $d_x=50$, while the constrained PF robustly tracks the system state, even for $d_x=400$, with just $N=100$ particles.

\begin{figure}[htb]
\centerline{
    \includegraphics[width=0.6\linewidth]{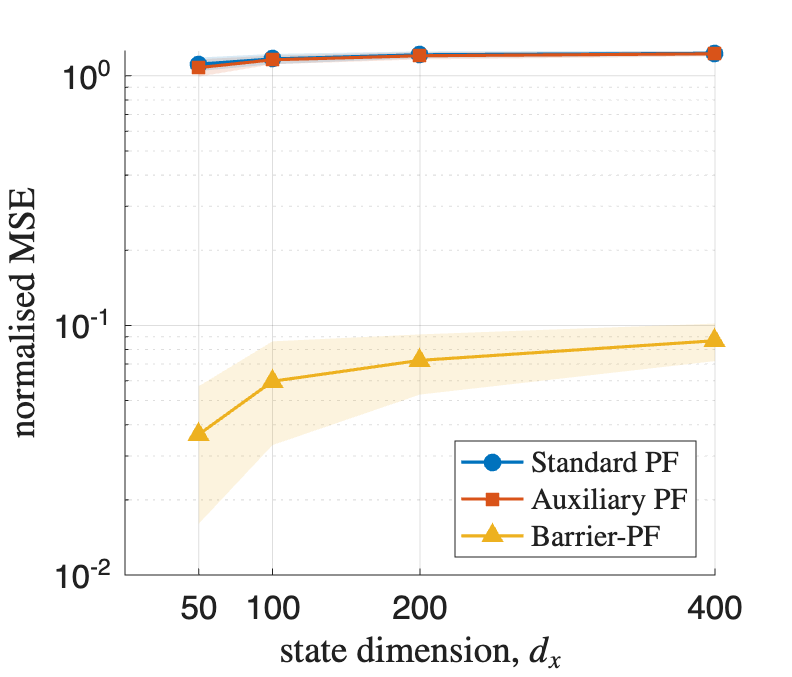}
}
\caption{Normalised MSE versus state dimension $d_x$. Each value in the plot is the average of 30 independent simulation runs. The shaded areas correspond to one standard deviation. The number of particles is fixed, $N=100$.} 
\label{fNMSE}
\end{figure}

Figure \ref{fPdfs} displays the marginal posterior densities of four state variables ($X_1$, $X_{10}$, $X_{50}$ and $X_{60}$) at time $t_{75}$ for one simulation of the constrained PF with $d_x=400$ and $N=100$. The pdfs are kernel density estimators computed from the particles using the matlab function {\tt ksdensity}. We have selected illustrative cases where variables are estimated very accurately ($X_1$ and $X_{50}$) and cases where the error is larger ($X_{10}$ and $X_{60}$). This simulation shows that the constrained PF remains locked to the system state and the particle sets have sufficient diversity to provide useful approximations of the marginal densities.

\begin{figure}[htb]
\centerline{
    \includegraphics[width=0.43\linewidth]{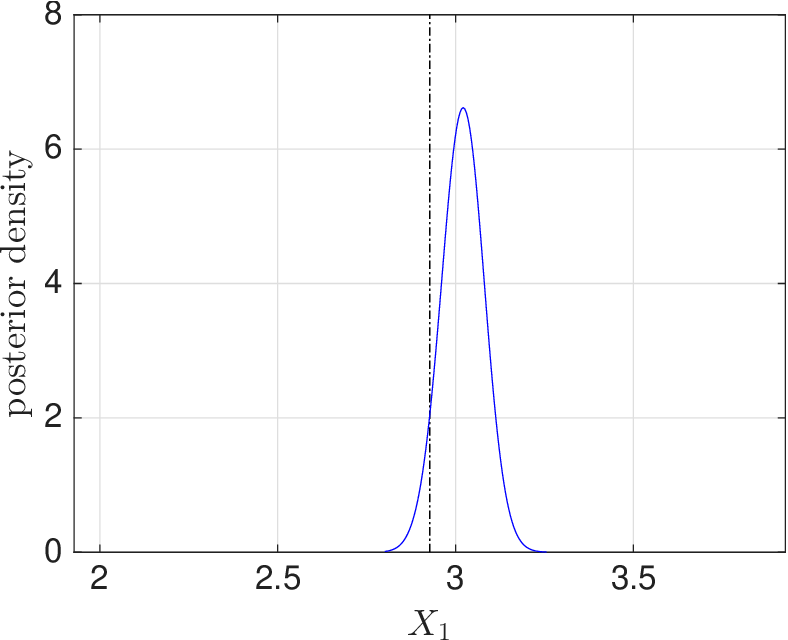}\quad
    \includegraphics[width=0.43\linewidth]{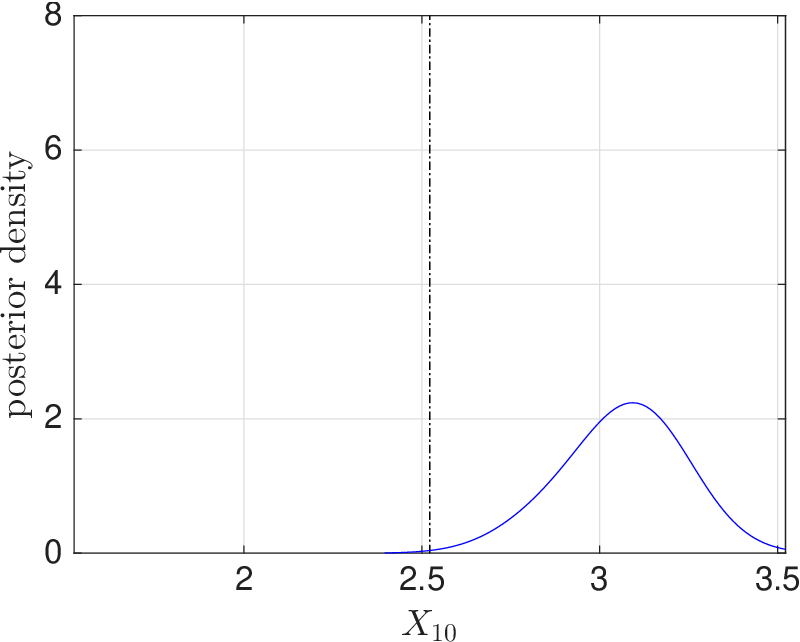}
}
\centerline{
    \includegraphics[width=0.43\linewidth]{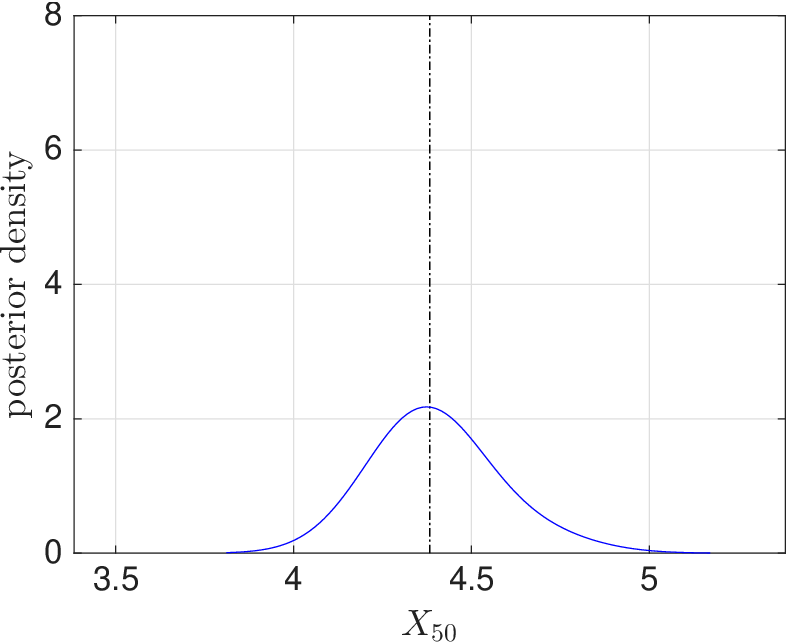}\quad
    \includegraphics[width=0.43\linewidth]{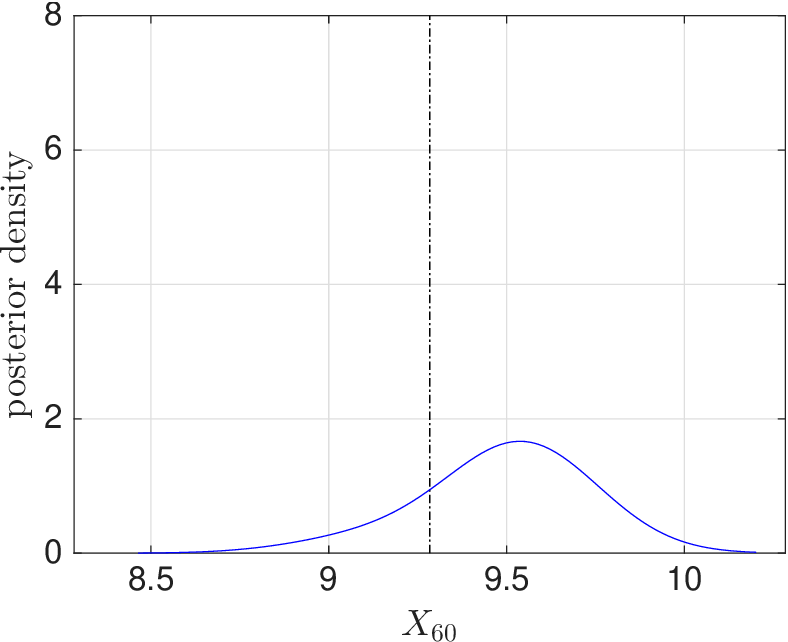}
}
\caption{Marginal posterior pdfs of four state variables ($X_1$, $X_{10}$, $X_{50}$ and $X_{60}$) computed with the constrained PF in one single simulation, with $d_x=400$ and $N=100$.}
\label{fPdfs}
\end{figure}

%
\section{Proofs} \label{sProofs}

%
\subsection{An ancillary result} \label{ssAncillary}
Let $\pi_n^N(\sd x)=\frac{1}{N}\sum_{i=1}^N \delta_{x_n^i}(\sd x)$ be the approximation of $\pi_n$ constructed by way of the standard particle filter.
 
\begin{lemma} \label{lmConvPF-2}
Let Assumptions \ref{assModel1} and \ref{assModel2} hold. If, in addition, there is some $\zeta>0$ such that $\inf_{n\ge 1} g_n > \zeta$, then, for any $p\ge 1$,
\beq
\sup_{n\ge 1, |f|\le 1} \| \pi_n(f) - \pi_n^N(f) \|_p \le \frac{\bar c}{\sqrt{N}}
\nn
\eeq
for some constant $\bar c<\infty$ independent of $N$.
\end{lemma}

\begin{proof} 
Let $f:\mX\mapsto\Real$ be a measurable test function, with $\|f\|_\infty<\infty$. The difference $\pi_n(f)-\pi_n^N(f)$ can be expanded as
\beqa
\pi_n^N(f)-\pi_n(f) &=& \sum_{i=0}^{n-1} \Phi_{n-i:n}\left( \pi_{n-i}^N \right)(f) - \Phi_{n-i:n}\left( \Phi_{n-i}\left( \pi_{n-i-1}^N \right) \right)(f) \nn\\
&& + \Phi_{0:n}\left( \pi_0^N \right)(f) - \Phi_{0:n}\left( \pi_0 \right)(f). \nn
\eeqa
and a straightforward application of Minkowski's inequality yields
\beqa
\| \pi_n^N(f)-\pi_n(f) \|_p &\le& \sum_{i=0}^{n-1} \left\| \Phi_{n-i:n}\left( \pi_{n-i}^N \right)(f) - \Phi_{n-i:n}\left( \Phi_{n-i}\left( \pi_{n-i-1}^N \right) \right)(f) \right\|_p \nn\\
&& + \left\| \Phi_{0:n}\left( \pi_0^N \right)(f) - \Phi_{0:n}\left( \pi_0 \right)(f) \right\|_p. \label{eql_pftelescope2}
\eeqa
for any $p\ge 1$. 

For $i=0$ in \eqref{eql_pftelescope2}, we have 
\beq
\left\| \Phi_{n:n}\left( \pi_{n-i}^N \right)(f) - \Phi_{n:n}\left( \Phi_{n}\left( \pi_{n-1}^N \right) \right)(f) \right\|_p = \| \pi_n^N(f)-\Phi_n(\pi_{n-1}^N)(f)\|_p
\eeq
and it is straightforward to show that
\beqa
\pi_n^N(f)-\Phi_n(\pi_{n-1}^N)(f) &=& \frac{
    \xi_n^N(g_nf)
}{
    \xi_n^N(g_nf)
} - \frac{
    K_n\pi_{n-1}^N(g_nf)
}{
    K_n\pi_{n-1}^N(g_nf)
} \pm \frac{
    K_n\pi_{n-1}^N(g_nf)
}{
    \xi_n^N(g_nf)
} \nn\\
&=& \frac{
    \xi_n^N(g_nf) - K_n\pi_{n-1}^N(g_nf)
}{
    \xi_n^N(g_n)
} \nn\\
&& + \frac{
        K_n\pi_{n-1}^N(g_nf)
}{
    K_n\pi_{n-1}^N(g_n)
}\frac{
    K_n\pi_{n-1}^N(g_n) - \xi_n^N(g_n)
}{
    \xi_n^N(g_n)
}.
\nn
\eeqa
Since $\left| 
    \frac{
        K_n\pi_{n-1}^N(g_nf)
    }{
        K_n\pi_{n-1}^N(g_n)
    }
\right| = | \Phi_n(\pi_{n-1}^N)(f) | \le \| f \|_\infty$, we readily see that
\beq
\left|
    \pi_n^N(f)-\Phi_n(\pi_{n-1}^N)(f)
\right| \le \frac{
    2\|f\|_\infty
}{
    \xi_n^N(g_n)
}\left|
    \xi_n^N(g_n) - K_n\pi_{n-1}^N(g_n)
\right|.
\nn
\eeq
From assumption $\inf_n g_n > \zeta$, we obtain $\xi_n^N(g_n)\ge \zeta>0$ and a straightforward application of the Marcinkiewicz-Zygmund inequality\footnote{Note that $K_n\pi_{n-1}^N(\sd x_n) = \frac{1}{N}\sum_{j=1}^N K_n(x_{n-1}^j,\sd x_n)$ and $\xi_n^N(\sd x_n) = \frac{1}{N}\sum_{j=1}^N \delta_{\bar x_n^j}(\sd x_n)$, where $\bar x_n^j \sim K_n(x_{n-1}^j,\sd x_n)$. If we let $\bar\mF_{n-1}$ be the $\sigma$-algebra generated by $\{x_{0:n-1}^{1:N}, \bar x_{1:n-1}^{1:N} \}$, it is clear that $\mbE[ f(\bar x_n^j) | \bar \mF_n ] = \int f(x_n) K_n(x_{n-1}^j,\sd x_n)$ for every $j=1, ..., N$.} yields
\beq
\left\|
    \pi_n^N(f)-\Phi_n(\pi_{n-1}^N)(f)
\right\|_p \le \frac{
    2c_p\|f\|_\infty\|g\|_\infty
}{
    \zeta \sqrt{N}
},
\label{eqlLp_error_n}
\eeq
where $c_p$ is a constant that depends only on $p$ (but is independent of $\xi_n^N$ and, in particular, of $N$) and we have used the bound $|g_n|\le\|g\|_\infty$ from Assumption \ref{assModel1}.

For $1 \le i \le n-1$ in \eqref{eql_pftelescope2}, let us first note that Assumption \ref{assModel2} leads to
\beqa
\dtv{\Phi_{n-i:n}(\pi_{n-i}^N) - \Phi_{n-i:n}(\Phi_{n-i}(\pi_{n-i-1}^N))} &\le&\\
r_i \dtv{\Phi_{n-i+1}(\pi_{n-i}^N) - \Phi_{n-i:n-i+1}(\pi_{n-i-1}^N)} &=&\nn \\
r_i \sup_{\|f\|_\infty \le 1} \frac{1}{2}\left|
    \Phi_{n-i+1}(\pi_{n-i}^N)(f) 
    - \Phi_{n-i+1}(\Phi_{n-i}(\pi_{n-i-1}^N))(f)
\right|.&&
\label{eqlContract2.0}
\eeqa
However, for any $j\ge 2$, one can show that
\beqa
\Phi_j(\pi_{j-1}^N)(f) &=& \frac{
    \int f(x_j)g_j(x_j) \int K_j(x_{j-1},\sd x_j) g_{j-1}(x_{j-1}) \xi_{j-1}^N(\sd x_{j-1})
}{
    \int g_j(x_j) \int K_j(x_{j-1},\sd x_j) g_{j-1}(x_{j-1}) \xi_{j-1}^N(\sd x_{j-1})
} \nn \\
&=& \frac{
    \xi_{j-1}^N\left( K_j(g_jf)g_{j-1} \right)
}{
    \xi_{j-1}^N\left( K_j(g_j)g_{j-1} \right)
} 
\label{eqlContract2.1}
\eeqa
and, similarly,
\beqa
\Phi_j(\Phi_{j-1}(\pi_{j-2}^N)(f) &=& \frac{
    \int f(x_j)g_j(x_j) \int K_j(x_{j-1},\sd x_j) g_{j-1}(x_{j-1}) K_{j-1}\pi_{j-2}^N(\sd x_{j-1})
}{
    \int g_j(x_j) \int K_j(x_{j-1},\sd x_j) g_{j-1}(x_{j-1}) K_{j-1}\pi_{j-2}^N(\sd x_{j-1})
} \nn \\
&=& \frac{
    K_{j-1}\pi_{j-2}^N\left( K_j(g_j f)g_{j-1} \right)
}{
    K_{j-1}\pi_{j-2}^N\left( K_j(g_j)g_{j-1} \right)
}, 
\label{eqlContract2.2}
\eeqa
where $K_j(g_j f)(x_{j-1}) = \int f(x_j) g_j(x_j) K_j(x_{j-1},\sd x_j)$. Combining \eqref{eqlContract2.1} and \eqref{eqlContract2.2} yields
\beqa
\Phi_j(\pi_{j-1}^N)(f) - \Phi_j(\Phi_{j-1}(\pi_{j-2}^N)(f) &=& \nn \\
\frac{
    \xi_{j-1}^N\left( K_j(g_jf)g_{j-1} \right)
}{
    \xi_{j-1}^N\left( K_j(g_j)g_{j-1} \right)
} - \frac{
    K_{j-1}\pi_{j-2}^N\left( K_j(g_j f)g_{j-1} \right)
}{
    K_{j-1}\pi_{j-2}^N\left( K_j(g_j)g_{j-1} \right)
} \pm \frac{
    K_{j-1}\pi_{j-2}^N\left( K_j(g_j f)g_{j-1} \right)
}{
    \xi_{j-1}^N\left( K_j(g_j)g_{j-1} \right)
} &=&\nn\\
\frac{
    \xi_{j-1}^N\left( K_j(g_jf)g_{j-1} \right) - K_{j-1}\pi_{j-2}^N\left( K_j(g_j f)g_{j-1} \right)
}{
    \xi_{j-1}^N\left( K_j(g_j)g_{j-1} \right)
} &&\nn\\
+ \rho_j(f) \frac{
    K_{j-1}\pi_{j-2}^N\left( K_j(g_j)g_{j-1} \right) - \xi_{j-1}^N\left( K_j(g_j)g_{j-1} \right) 
}{
    \xi_{j-1}^N\left( K_j(g_j)g_{j-1} \right)
} &&
\label{eqlContract2.3}
\eeqa
where 
\beq
\rho_j(\sd x_j) = \frac{
    g_j(x_j) \int K_j(x_{j-1},\sd x_j) g_{j-1}(x_{j-1}) K_{j-1}\pi_{j-2}^N(\sd x_{j-1})
}{
    \int g_j(x_j) \int K_j(x_{j-1},\sd x_j) g_{j-1}(x_{j-1}) K_{j-1}\pi_{j-2}^N(\sd x_{j-1})
}
\nn
\eeq
is a probability measure. From \eqref{eqlContract2.3} and the fact that $|\rho_j(f)|\le \|f\|_\infty$ we readily find the inequality
\beqa
\left|
    \Phi_j(\pi_{j-1}^N)(f) - \Phi_j(\Phi_{j-1}(\pi_{j-2}^N)(f)
\right| &\le& \nn\\
\frac{
    2\|f\|_\infty \left|
     \xi_{j-1}^N\left( K_j(g_j)g_{j-1} \right) - K_{j-1}\pi_{j-2}^N\left( K_j(g_j)g_{j-1} \right)
\right|
}{
    \xi_{j-1}^N\left( K_j(g_j)g_{j-1} \right)
} &&
\label{eqlContract2.4}
\eeqa
and, combining \eqref{eqlContract2.0} and \eqref{eqlContract2.4} with $j=n-i+1$ yields
\beqa
\sup_{
    |f| \le 1
} \left|
    \Phi_{n-i:n}(\pi_{n-i}^N)(f) - \Phi_{n-i:n}(\Phi_{n-i}(\pi_{n-i-1}^N))(f)
\right| &\le& \nn \\
2r_i
\frac{
    \left|
        \xi_{n-i}^N\left( K_{n-i+1}(g_{n-i+1})g_{n-i} \right) - K_{n-i}\pi_{n-i-1}^N\left( K_{n-i+1}(g_{n-i+1})g_{n-i} \right)
    \right|
}{
    \xi_{n-i}^N\left( K_{n-i+1}(g_{n-i+1})g_{n-i} \right)
}.
\label{eqlContract2.5}
\eeqa
Recalling, again, the uniform bounds on the likelihoods, $\zeta < g_n < \| g \|_\infty$, we realise that
\beq
\zeta^2 < K_{n-i+1}(g_{n-i+1})g_{n-i} < \|g\|^2.
\eeq
Then, combining \eqref{eqlContract2.5} with the Marcinkiewicz-Zygmund inequality, applied in the same way as in \eqref{eqlLp_error_n}, we readily arrive at the $L^p$ bound
\beq
\left\|
    \Phi_{n-i:n}(\pi_{n-i}^N)(f) - \Phi_{n-i:n}(\Phi_{n-i}(\pi_{n-i-1}^N))(f)
\right\|_p \le \frac{
    2 r_i c_p \|g\|_\infty^2
}{
    \zeta^2 \sqrt{N}
}.
\label{eqlLp_error_ni}
\eeq
that holds for any real measurable test function $f:\mX\mapsto\Real$ with $\|f\|_\infty\le 1$.

The calculation for the last term on the right hand side of \eqref{eql_pftelescope2} is very similar. Assumption \ref{assModel2} yields
\beq
\sup_{|f|\le 1} \frac{1}{2} \left| 
    \Phi_{0:n}(\pi_0^N)(f)-\Phi_{0:n}(\pi_0)(f)
\right| \le r_n \sup_{|f|\le 1} \frac{1}{2} \left| 
    \Phi_1(\pi_0^N)(f)-\Phi_1(\pi_0)(f)
\right|
\label{eqlContract2.6}
\eeq
and 
\beqa
\left| 
    \Phi_1(\pi_0^N)(f)-\Phi_1(\pi_0)(f)
\right| &\le&\nn\\
\left|
    \frac{
        K_1\pi_0^N(g_1f) - K_1\pi_0(g_1f)
    }{
        K_1\pi_0(g_1)
    } + K_1\pi_0(f) \frac{
        K_1\pi_0(g_1) - K_1\pi_0^N(g_1)
    }{
        K_1\pi_0(g_1)
    }
\right| &\le&\nn \\
\frac{
    2\|f\|_\infty
}{
    K_1\pi_0(g_1)
} \left|
    K_1\pi_0^N(g_1) - K_1\pi_0(g_1)
\right|.
\label{eqlContract2.7}
\eeqa
Using again the bounds $\zeta < g_1 < \|g\|_\infty$ and the Marcinkiewicz-Zygmund inequality, we arrive at
\beq
\left\| 
    \Phi_1(\pi_0^N)(f)-\Phi_1(\pi_0)(f)
\right\|_p \le \frac{
    2\|f\|_\infty\|g\|_\infty c_p
}{
    \zeta \sqrt{N}
},
\label{eqlContract2.8}
\eeq
where the constant $c_p<\infty$ depends only on $p\ge 1$. Hence, for any measurable and real test function $f$ with $\|f\|_\infty \le 1$, the combination of \eqref{eqlContract2.6}, \eqref{eqlContract2.7} and \eqref{eqlContract2.8} yields
\beq
\left\| 
    \Phi_{0:n}(\pi_0^N)(f)-\Phi_{0:n}(\pi_0)(f)
\right\|_p \le \frac{
    4 r_n c_p\|g\|_\infty
}{
    \zeta \sqrt{N}
}.
\label{eqlLp_error_0}
\eeq

Finally, recalling Assumption \ref{assModel2} (in particular, $r_0=1$ and $\sum_{i=0}^\infty r_i < \infty$) and taking together \eqref{eql_pftelescope2}, \eqref{eqlLp_error_n}, \eqref{eqlLp_error_ni} and \eqref{eqlLp_error_0} we arrive at
\beq
\sup_{|f| \le 1} \| \pi_n^N(f) - \pi_n(f) \|_p \le \frac{
    4 c_p \|g\|_\infty^2 \sum_{i=0}^\infty r_i
}{
    \zeta^2
} \times \frac{
    1
}{
    \sqrt{N}
}. 
\nn
\eeq
\qed
\end{proof}

%
\subsection{Proof of Theorem \ref{thCPFs}} \label{ssCPFs}

For the first part of the theorem, let us take the triangle inequality 
\beq
\|\pi_n(f)-\pi_n^{l,N}(f)\|_p \le |\pi_n(f)-\pi_n^l(f)| + \|\pi_n^l(f)-\pi_n^{l,N}(f)\|_p.
\label{eqpf_proof1}
\eeq 
Assumptions \ref{assCl-2} and \ref{assModel2}, and the uniform bounds $\zeta < g_n < \|g\|_\infty$ on the likelihoods yield, by way of Theorem \ref{thUniformDTV-1}, the first term on the r.h.s. of \eqref{eqthCPF1}. In particular, there is some constant $c<\infty$ such that 
\beq
\sup_{|f|\le 1} |\pi_n(f)-\pi_n^l(f)| \le c \frac{1-\varepsilon^l}{\varepsilon^l}.
\label{eqpf_proof2}
\eeq
As for the term $\|\pi_n^l(f)-\pi_n^{l,N}(f)\|_p$ on the right-hand side of \eqref{eqpf_proof1}, we note that Assumption \ref{assModel1} and \eqref{eqThCPFsAsC} imply that 
\beq
\frac{k_n(x',x)}{k_n(x'',x)} \ge \gamma_l := \bar\gamma_l\|k\|_\infty^{-1}
\label{eqForPierre}
\eeq
for any $x',x'' \in \mC_{n-1}^l$ and any $x\in\mC_n^l$. The inequality \eqref{eqForPierre} together with the assumption $\inf_{n \ge 1} g_n \ge \zeta > 0$ enables us to use Theorem 7.4.4 in \cite{DelMoral04}, which yields
\beq
\sup_{|f|\le 1, n\ge 0} \|\pi_n^l(f)-\pi_n^{l,N}(f)\|_p \le \frac{
	c_p \|g\|_\infty
}{
	\zeta\gamma_l^3 \sqrt{N}
}
\label{eqpf_proof3}
\eeq
for some $c_p<\infty$ that depends only on $p \ge 1$. Taking \eqref{eqpf_proof1}, \eqref{eqpf_proof2} and \eqref{eqpf_proof3} together yields inequality \eqref{eqthCPF1} in the statement of Theorem \ref{thCPFs}.

Inequality \eqref{eqthCPF2} is straightforward. Simply note that, for fixed $l$, $\gamma_l>0$ and then let $N\to\infty$ in \eqref{eqthCPF1}.

For the third part of the theorem, we iterate the triangle inequality to arrive at
\beqa
\lim_{l\to\infty} \| \pi_n(f) - \pi_n^{l,N}(f) \|_p &\le& \| \pi_n(f) - \pi_n^N(f) \|_p  +  \lim_{l\to \infty} \| \pi_n^N(f) - \pi_n^l(f) \|_p \nn \\
&& + \| \pi_n^l(f) - \pi_n^{l,N}(f) \|_p,
\label{eqpf_proof3.5}
\eeqa
and then we note that $\sup_{l\ge 1} |\pi_n^l(f)| \vee |\pi_n^{l,N}(f)| \le \|f\|_\infty < \infty$, hence Lebesgue's dominated convergence theorem yields
\beq 
\lim_{l\to\infty} \| \pi_n^N(f) - \pi_n^l(f) \|_p = \| \pi_n^N(f) - \lim_{l\to\infty} \pi_n^l(f) \|_p = \| \pi_n^N(f) - \pi_n(f) \|_p 
\label{eqpf_proof4}
\eeq
and, similarly,
\beq
\lim_{l\to\infty} \| \pi_n^l(f) - \pi_n^{l,N}(f) \|_p = \left\|
	\lim_{l\to\infty} \pi_n^l(f) - \pi_n^{l,N}(f) 
\right\|_p = \| \pi_n(f) - \pi_n^N(f) \|_p.
\label{eqpf_proof5}
\eeq
Taking together \eqref{eqpf_proof3.5}, \eqref{eqpf_proof4} and \eqref{eqpf_proof5} we obtain
\beq
\lim_{l\to\infty} \| \pi_n(f) - \pi_n^{l,N}(f) \|_p \le 3 \| \pi_n(f)-\pi_n^N(f) \|_p
\nn
\eeq
and Lemma \ref{lmConvPF-2} ensures the existence of some constant $\bar c<\infty$ such that 
$$
\sup_{|f|\le 1, n\ge 0} \| \pi_n(f)-\pi_n^N(f) \| \le \frac{\bar c}{\sqrt{N}}.
$$
Hence, \eqref{eqthCPF3} is satisfied. \quad $\QED$

%
\subsection{Proof of Theorem \ref{thMarginals}} \label{ssThDiscretisation}

We follow an induction argument. At time $n=0$ we trivially have $\pi_0=\pi_0^h$, hence $\pi_0(f)=\pi_0^h(f)$. The induction hypothesis is
\beq
| \pi_{n-1}(f) - \pi^h_{n-1}(f) | \le C_{n-1} h
\label{eqP10}
\eeq
for some $n < M$, where $f \in C_B^4(\Real^{d_x})$ is a test function such that $\| f \|_\infty \le 1$. 

We readily see that
\beq
| \xi^h_n(f) - \xi_n(f) | = \left| K_n^h\pi^h_{n-1}(f) -  K_n\pi_{n-1}(f) \right| 
= \left|
	\pi^h_{n-1}( \bar f_n^h ) - \pi_{n-1}( \bar f_n )
\right|, \label{eqP11}
\eeq
where 
\beq
\bar  f_n(x_{n-1}) = \int  f(x_n)  K_n(x_{n-1},\sd x_n)
\quad \text{and} \quad
\bar  f^h_n(x_{n-1}) = \int  f(x_n)  K_n^h(x_{n-1},\sd x_n), \nn
\eeq
hence a triangle inequality yields
\beq
\left|
	\pi^h_{n-1}( \bar f_n^h ) - \pi_{n-1}( \bar f_n )
\right| \le \left| 
	\pi^h_{n-1}(\bar f_n^h) - \pi_{n-1}^h(\bar f_n) 
\right| + \left|
	\pi_{n-1}^h(\bar f_n) - \pi_{n-1}( \bar f_n )
\right|.
\label{eqP12}
\eeq

It is straightforward to show that $\| f\|_\infty \le 1$ implies $\| \bar f_n \|_\infty \le 1$. Moreover, assumption (ii) in the statement of Theorem \ref{thMarginals} implies that 
\beq
\sup_{n \le M, |\alpha|\le 4} \| \bar  f_n^{(\alpha)} \|_\infty < \infty;
\nn
\eeq 
hence, $\bar  f_n \in C_B^4(\Real^{d_x})$ and $\|\bar f_n\|_\infty\le 1$. Therefore, we can apply the induction hypothesis \eqref{eqP10} to obtain
\beq
\left|
	\bar \pi^h_{n-1}(\bar f_n) - \pi_{n-1}( \bar f_n )
\right| = C_{n-1}h,
\label{eqP12.5}
\eeq
which accounts for the second term on the right-hand side of \eqref{eqP12}. For the first term on the right-hand side of \eqref{eqP12}, we note that Assumption \ref{asW} yields
\beq
\sup_{x\in \Real^{d_x}} | \bar f^h_n(x) - \bar f_n(x) | \le \check C_n h
\nn
\eeq
for some $\check C<\infty$ and, therefore,
\beq
\left| 
	\pi^h_{n-1}(\bar f_n^h) - \bar \pi^h_{n-1}(\bar f_n) 
\right| = \left| 
	\pi^h_{n-1}( \bar f_n^h - \bar f_n ) 
\right| \le \check C_n h.
\label{eqP13}
\eeq
Substituting \eqref{eqP13} and \eqref{eqP12.5} in \eqref{eqP12}, and then \eqref{eqP12} back in \eqref{eqP11}, yields
\beq
\left| \xi_n^h(f) - \xi_n(f) \right| \le \tilde C_n h,
\label{eqP14}
\eeq
where $\tilde C_n = C_{n-1} + \check C_n < \infty$.

Next, we write the difference $\pi^h_n( f) - \pi_n( f)$ in terms of $\xi^h_n$ and $\xi_n$ as
\beqa
\pi^h_n( f) - \pi_n( f) &=& \frac{
	\xi_n^h(g_n f)
}{
	\xi_n^h(g_n)
} - \frac{
	\xi_n(g_n f)
}{
	\xi_n(g_n)
} \pm \frac{
	\xi_n^h(g_n f)
}{
	\xi_n(g_n)
} \nn\\
&=& \frac{
	\xi_n^h(g_n f) - \xi_n(g_n f)
}{
	\xi_n(g_n)
} + \frac{
	\xi_n^h(g_n f)
}{
	\xi_n^h(g_n)
} \times \frac{
	\xi_n(g_n) - \xi_n^h(g_n)
}{
	\xi_n(g_n)
}, \nn
\eeqa
which readily yields the bound
\beq
\left| \pi^h_n(f) - \pi_n(f) \right| \le
\frac{
	\left|
		\xi_n^h(g_n f) - \xi_n(g_n f)
	\right|
}{
	\xi_n(g_n)
} + \|  f \|_\infty \frac{
	\left|
		\xi_n(g_n) - \xi_n^h(g_n)
	\right|
}{
	\xi_n(g_n)
}.
\label{eqP15}
\eeq
The difference $\left|
	\xi_n^h(g_n f) - \xi_n(g_n f)
\right|$ can be upper bounded as 
\beq 
\left|
	\xi_n^h(g_n f) - \xi_n(g_n f)
\right| \le \| f\|_\infty \left|
	\xi_n(g_n) - \xi_n^h(g_n)
\right|
\label{eqP15.5}
\eeq
and, since $g_n \in C_B^4(\Real^{d_x})$ and $\|g_n\|_\infty \le 1$, the inequality \eqref{eqP14} implies that
\beq
\left|
	\xi_n(g_n) - \xi_n^h(g_n)
\right| \le \tilde C_n h.
\label{eqP16}
\eeq
Combining \eqref{eqP16} and \eqref{eqP15.5} back into \eqref{eqP15} we obtain the bound
\beq
\left| \pi^h_n( f) - \pi_n( f) \right| \le \frac{
	2 \| f\|_\infty \tilde C_n h
}{
	\xi_n(g_n)
} \label{eqP17}
\eeq
which holds for all $f \in C_B^4(\Real^{d_x})$. Since $g_n>0$ (and $\|  f \|_\infty \le 1$), inequality \eqref{eqP17} yields
$
\left| \pi_n^h(f) - \pi_k(f) \right| \le C_n h
\nn
$, where $C_n = \frac{
	2 \tilde C_n h
}{
	\xi_n(g_n)
} < \infty$. \hfill $\QED$

%
\section{Conclusions} \label{sConclusions}

We have introduced a class of constrained PFs for continuous-discrete Bayesian filtering problems, where the state of the model is an It\^o diffusion and the (partial and noisy) observations are collected at discrete times. The constraint imposes a compact support on the system state at each observation time, which restricts the exploration of the PF to ``typical'' (high likelihood) solutions. Under regularity assumptions, we have analytically obtained approximation error rates for both the exact constrained filter and its PF implementation. We have also extended this analysis to obtain error rates that incorporate the effect of the numerical discretisation scheme employed to simulate the underlying It\^o diffusion. Finally, we have presented numerical results that illustrate how a simple (approximate) numerical implementation of the constraint yields a drastic performance improvement in the performance of a standard PF on a stochastic Lorenz 96 model.

\begin{acknowledgement}
JM acknowledges the financial support of grant PID2024-158181NB-I00 NISA, funded by MCIN/AEI/10.13039/501100011033 and ERDF, and by Community of Madrid, under Grant IDEA-CM (TEC-
2024/COM-89). JM thanks the mathematical research institute MATRIX, in Australia, where part of this research was performed. 
\end{acknowledgement}

\bigskip
%
\bibliographystyle{spmpsci}
\bibliography{bibliografia.bib}

\end{document}